\documentclass[preprint]{elsarticle}

\usepackage{amssymb}
\usepackage{amsthm}
\usepackage{amsmath}
\usepackage{bm}
\usepackage{soul}
\usepackage{multirow}
\usepackage{subcaption}
\usepackage{natbib}
\usepackage[usenames, dvipsnames]{color}

\usepackage{tikz,xcolor}

\begin{document}

\begin{frontmatter}

%\title{Joint ruin probability for a two-dimensional risk process with claims being a mixture of two exponentials
%}

\title{Ruin probability for the quota share model with~phase-type distributed claims
}

\author[label]{Krzysztof Burnecki}
\address[label]{Faculty of Pure and Applied Mathematics \\ Hugo Steinhaus Center, Wroc{\l}aw University of Science and Technology\\ Wybrze{\.z}e Wyspia{\'n}skiego 27, 50-370 Wroc{\l}aw, Poland}
%\ead{krzysztof.burnecki@pwr.edu.pl}

\author[label]{Zbigniew Palmowski}

\author[label]{Marek Teuerle\corref{cor}}
\ead{marek.teuerle@pwr.edu.pl}
%\cortext[cor]{Corresponding author}

\author[label]{Aleksandra Wilkowska}
%\ead{aleksandra.wilkowska@pwr.edu.pl}

\begin{abstract}
In this paper, we generalise the results presented in the literature 
for the ruin probability for the insurer--reinsurer model under a pro-rata reinsurance contract. We consider claim amounts that are described by a phase-type distribution that includes exponential, mixture of exponential, Erlang, and mixture of Erlang distributions.  We derive the ruin probability formulas with the use of change-of-measure technique and present important special cases. We illustrate the usefulness of the introduced model by fitting it to the real-world loss data. With the use of statistical tests and graphical tools, we show that the mixture of Erlangs is well-fitted to the data and is superior to other considered distributions. This justifies the fact that the presented results can be useful in the context of risk assessment of co-operating insurance companies.
%in which the insurer and reinsurer share premiums and losses according to a fixed percentage. 
\end{abstract}

\begin{keyword}
multidimensional risk process \sep ruin probability \sep change of~measure \sep phase-type distribution \sep mixture of Erlang distributions
\end{keyword}

\end{frontmatter}

\newtheorem{theorem}{\bf Theorem}[section]
\newtheorem{condition}{\bf Condition}[section]
\newtheorem{corollary}{\bf Corollary}[section]
\newtheorem{definition}{\bf Definition}[section]
\newtheorem{example}{\bf Example}[section]
\newtheorem{lemma}{\bf Lemma}[section]
\newcommand{\sgn}{\mathrm{{sgn}}}
\newcommand{\Prob}{\mathrm{{P}}}
\newcommand{\E}{\mathrm{E}}

% \linenumbers

\section{Introduction}

Risk theory in general and ruin probabilities in particular 
have been an active area of
research since the classical Cramér-Lundberg model, introduced in 1903 by the Swedish actuary Filip Lundberg \citep{lund} and then generalised in the 1930's  by Harald Cram\'er \citep{grenander95}.  

The Cramér-Lundberg model describes the surplus of an insurance company that experiences two opposing cash flows: incoming premiums and outgoing claims. The traditional approach in risk theory is to study the
probability of ruin, that is, the probability that the risk process 
will ever go below zero \citep{asmalb10}.
Ruin is considered a technical term. It does not mean that
the company becomes bankrupt. If ruin occurs, this is interpreted to mean that the company has to take action to make the business profitable. For solvency purposes, the probability of ruin can be used as a rough approximation of the insolvency.
Moreover, setting it to an acceptably low level, the needed initial capital and the rate of premiums can be estimated. It can also serve as a useful tool in long-range planning for the use of insurer’s funds. In addition, ruin theory has deep methodological links and applications to other fields of applied probability, such as queueing theory and mathematical finance \citep{asmalb10}.

The ruin probabilities in infinite and finite time, even for the classical risk process, can only be calculated for a few special cases of the claim amount distribution. For the infinite horizon case, there are well-known elementary results for zero initial capital, and the exponential and mixture of two exponential claim amount distributions,
see \cite{panjer1992,gra91}. For the results for general phase-type distributions, in particular for mixture of $n$ exponential distributions, see \cite{,asmalb10}. For the finite-horizon case, the only convenient "semi-elementary" formula (involving only a simple integral) exists for the exponential distribution \citep{rolski,gra91}. However, this case can always be approximated by the Monte Carlo method.

Recently, multidimensional risk processes have been introduced in the literature to account for multiple lines of business of an insurance company and collaborating insurance companies. The ruin probability can be now defined in several ways, e.g.
when all lines or all companies are ruined or at least one.  The multidimensional ruin problem for light-tailed
claims and general ruin sets was studied for the first time in \cite{collamore1996} and multidimensional heavy-tailed processes in \cite{hultetal2005}. They mainly concentrated on multivariate regularly varying random walks and calculated sharp boundaries for the asymptotic ruin probability.

Since different risks usually have an effect on a few lines of business at the same time, the statistical dependence among claims in these lines should be taken into account. The multidimensional risk process was specialised to the two-dimensional case with claims shared with a predetermined proportion in \cite{app08,app08b}. This case is usually referred to as the insurer--reinsurer model, as it well describes the quota share proportional treaty. It can also be used to model two branches of the same insurance
company. The ruin occurs here if one or both companies go bankrupt. The former case, which is more interesting from a practical point of view, is usually analysed, and the latter can be obtained from the former in a straightforward way. The only simple ruin probability
formulas for the insurer--reinsurer model were provided for exponentially distributed claims in \cite{app08b} (by explicitly inverting the Laplace transform) and later in \cite{burteuwilrisks21} (by means of a change-of-measure technique).  
Another type of dependence was studied in \cite{Behme2020}, where the link was established by a random bipartite network. An extension to a system of two insurers, where the first insurer is experiences claims arising from two independent compound Poisson processes and the second insurer covers a
proportion of the claims was introduced in \cite{Badescu2011}. In \cite{Michna2020}, a model driven
by a general spectrally positive or negative Lévy process was investigated, see also \cite{app08}.

In this paper, we derive the results for the infinite-time ruin probability for the general phase-type distributions. The article is organised as follows. In Section \ref{sec:model}, the model is presented and ruin probabilities are defined. In Section \ref{sec:phase} the results for phase-type claims for the classical Cramér-Lundberg model
 are recalled. The main results are presented in Section \ref{sec:two}. For the insurer--reinsurer model driven by the renewal process, we derive a ruin probability formula for the infinite-time horizon.
The special cases of mixture of exponential and Erlang distributions are presented.
In order to illustrate the usefulness of phase-type distributions in the context of ruin probability, in Section \ref{sec:sim} we analyze loss data from a Polish insurance company. We identify and validate the aggregate non-homogeneous Poisson by means of rigorous statistical tests and visual techniques. We show that the mixture of two Erlang distributions outperforms other considered distributions. This justifies the usefulness of the obtained results and importance of the mixture of two Erlang distributions in modelling the loss data. Section \ref{sec:con} summarises our results.

% The paper is structured as follows. In Section~2 we introduce a two-node insurance network.  In Section 3 we present and derive properties of the discussed class of claim amount distributions, namely the mixture of Erlangs (MOEs). Section~4 contains main results of the paper. We derive formulas for the join ruin probability  within the introduced two-node insurance network for the MOE claims. To this end we apply the change of measure technique. We also present easy-to-use formulas for two special cases of the MOEs, namely the mixture of two exponentials and Erlang. Those cases are numerically illustrated in Section~5. Concluding remarks are formulated in Section~\ref{sec:conclusions}.

\section{Insurer--reinsurer model}
\label{sec:model}

We consider here an insurance network that describes capitals of insurer and reinsurer companies that share a quota-share reinsurance contract. We assume that both the insurer and reinsurer participate in settling claims that have common origin. Formally, we can define the network  on the usual probability space $(\Omega,\mathcal{F},\Prob)$ as a system $\left(R_1(t),R_2(t)\right)_{t\geq 0}$ of two Cram\'er-Lundberg models in the following form:
\begin{align}
% \left( \!\!\!
% \begin{array}{c}
% R_1(t) \\
% R_2(t)
% \end{array}\!\!\!
% \right)=
% \left( \!\!\!
% \begin{array}{c}
% x_1 \\
% x_2
% \end{array}\!\!\!
% \right)
% +
% \left( \!\!\!
% \begin{array}{c}
% p_1\\
% p_2
% \end{array}\!\!\!
% \right) t
% -\left( \!\!\!
% \begin{array}{c}
% \delta  \\
% 1-\delta
% \end{array}\!\!\!
% \right)\sum_{i=1}^{N_{1} (t)} X_{1i}
% -\left(\!\!\!
% \begin{array}{c}
% \delta  \\
% 1-\delta
% \end{array}\!\!\!
% \right)\sum_{i=1}^{N_{2} (t)} X_{2i}.
\left( \!\!\!
\begin{array}{c}
R_1(t) \\
R_2(t)
\end{array}\!\!\!
\right)=
\left( \!\!\!
\begin{array}{c}
x_1 \\
x_2
\end{array}\!\!\!
\right)
+
\left( \!\!\!
\begin{array}{c}
p_1\\
p_2
\end{array}\!\!\!
\right) t
-\left( \!\!\!
\begin{array}{c}
\delta  \\
1-\delta
\end{array}\!\!\!
\right)\sum_{i=1}^{N(t)} X_{i}.
\label{2risk}
\end{align}
Here, $x_1$ and $x_2$ denote the initial capitals of the first and second reinsurer, $p_1$ and $p_2$ are their premium income rates, $\left(N(t)\right)_{t\geq 0}$ % and $\left(N_2(t)\right)_{t\geq 0}$ are independent 
 is a claim counting Poisson processes with %intensities $\lambda_1>0$ and $\lambda_2>0$
intensity $\lambda>0$ that is independent of claim amount sequence $\{X_{i}\}_{i\geq 1}$. Parameter $\delta\in[0;1]$ defines the split proportions $(\delta, 1-\delta)$  for the insurer and reinsurer, respectively.

Usually, we assume that
 \begin{align}
 \label{premiums}
 p_1&=(1+\theta_1)\lambda\delta\, \mathrm{E}(X_i), \nonumber\\
 p_2&=(1+\theta_2)\lambda (1-\delta)\,\mathrm{E}(X_i),
 \end{align}
 where $\theta_1,\theta_2 >0$ are the relative safety loadings. Due to higher acquisition and administration costs of the insurer, it is natural to assume that the premium rate for the insurer is higher than for the reinsurer and therefore the following relation holds:  $\theta_1>\theta_2$.

Our main goal is to obtain an analytical expression for the ruin probability for at least one of considered insurance companies in the infinite time horizon, which is formally defined as follows
\begin{equation}
	\psi_{OR}(u_1,u_2)=\mathbb{P}(\tau(u_1,u_2)< \infty ),
	\label{ruinOr}
\end{equation}
where $\tau_{OR}(u_1,u_2)$ is the ruin time:
\begin{equation}
\tau_{OR}(u_1,u_2)=\inf\{t\geq 0 : R_1(t) < 0 \vee R_2(t) < 0\}.
\label{stoppingOr}
\end{equation}
One can also be interested in the ruin probability for both companies at the same time in the infinite time horizon:
\begin{equation}
	\psi_{SIM}(u_1,u_2)=\mathbb{P}(\tau_{SIM}(u_1,u_2)< \infty ),
	\label{ruinSim}
\end{equation}
Here, the ruin time $\tau_{SIM}(u_1,u_2)$ is defined as follows:
\begin{equation}
\tau_{SIM}(u_1,u_2)=\inf\{t\geq 0 : R_1(t) < 0 \wedge R_2(t) < 0\}.
\label{stoppingSim}
\end{equation}

%We aim to establish the analytical result for the ruin probability by using purely probabilistic arguments.
Let us observe that in fact the ruin probabilities in infinite time (\ref{ruinOr}), (\ref{ruinSim}) of the risk process (\ref{2risk}) is the same as for the re-scaled process $\left(U_1(t),U_2(t)\right)_{t\geq 0}:=\left(R_1(t)/\delta,R_2(t)/(1-\delta)\right)_{t\geq 0}$, that is
\begin{align}
\left( \!\!\!
\begin{array}{c}
U_1(t) \\
U_2(t)
\end{array}\!\!\!
\right)=
\left( \!\!\!
\begin{array}{c}
u_1 \\
u_2
\end{array}\!\!\!
\right)
+
\left( \!\!\!
\begin{array}{c}
c_1\\
c_2
\end{array}\!\!\!
\right) t -
\sum_{i=1}^{N(t)} X_{i},
\label{2riskNormalized}
\end{align}
where $u_1$ and $u_2$ are equal to $x_1/\delta$ and $x_2/(1-\delta)$, and $c_1$ and $c_2$ are equal to $p_1/\delta$ and $p_2/(1-\delta)$, respectively.

% Our results generalize an analytical solution obtained in Ref. \cite{app08,app08b}, where the explicit result was only obtained in the case of claims having an exponential distribution. Moreover, our result for the exponentially distributed claims simplifies Theorem 2 from \cite{app08}. It is also worth noticing that Theorem 2 in Ref. \cite{app08} was obtained with the use of the double Laplace transform inversion technique. We derive an analytical result for the ruin probability by using purely probabilistic arguments. Our formula is also more computationally stable. To the best of our knowledge, in the literature, there are no other explicit solutions for the exact ruin probability in the infinite time horizon for the model considered here.

\section{Ruin probability for phase-type claims}
\label{sec:phase}
% Analysis performed in the previous sections for the claims being a mixture of Erlang distributions could be further generalized to more general phase-type distributions. Unfortunately, then the expression for the ruin probability is much less explicit. Still it can be used for numerical analysis as we show later.

We recall that
in \cite{app08} we find that for (ii) the following holds:
\begin{align}\label{masterformula}
\psi(u_1,u_2)=1-\int_{0}^{\infty}(1-{\psi}_2(z)){\mathbb{P}_{(u_1,T)}}(\mathrm{d}z),
\end{align}
where ${\psi}_2(z)$ is the ruin probability in infinite time for $U_2(t)$ with initial capital $z$ and
\begin{align}
\label{wzorP}
{\mathrm{P}_{(u_1,T)}}(\mathrm{d}z)=\!\mathrm{P}\left(\inf_{ s\leq T}U_1(s)>0, U_1(T) \in \mathrm{d}z | {U_1(0)=u_1} \right),
\end{align}
with the specific time point T such that
\begin{align}\label{T}
T=\frac{u_1-u_2}{(\theta_1-\theta_2)\lambda\, \E X_1}.
\end{align}
In this section we assume that the distribution $F$ of the generic claim size $X$
appearing in \eqref{2risk} is given by
\[F(x)=1-{\bm{\alpha}}e^{\mathbf{Q}x}\mathbf{1},\]
where $\mathbf{1}$ is a column vector with all its entries equal to $1$, $\bm{\alpha}$ is an initial distribution
of a continuous-time Markov chain on $m<\infty$ states
with a transition sub-rate matrix $\mathbf{Q}$ of dimension $m$.
We assume that this Markov chain is transient,
that is, $\mathbf{t}=-\mathbf{Q}\mathbf{1}\geq 0$ has a positive entry.
Then $X$ describes the lifetime of our Markov chain
and its density equals
\begin{equation}\label{densityP}
f(x)={\bm{\alpha}}e^{\mathbf{Q}x}\mathbf{t},\end{equation}
where $e^{\mathbf{A}}=\sum_{k=0}^\infty\frac{\mathbf{A}^k}{k!}$ for any
matrix $\mathbf{A}$.
If for example $X$ has exponential distribution with parameter $\beta$ then $m=1$ and
$\mathbf{Q}=-\beta$, $\mathbf{t}=\beta$ and $\bm{\alpha}=1$.

From Cor. 3.1 on p. 264 of \cite{asmussen} we have the following lemma that defines the ruin probability for one-dimensional risk process with claims being a phase-type distributed $(\bm{\alpha},\bm{Q})$.
\begin{lemma}
\begin{equation}\label{ruonphasetype}
\psi_2(z)=\bm{\alpha}_+e^{\mathbf{Q}_+ z}\mathbf{1},
\end{equation}
where
\begin{equation}\label{qplus}
\mathbf{Q}_+:=\mathbf{Q}+\mathbf{t}\bm{\alpha}_+\quad\text{and}\quad
\bm{\alpha}_+:=-\frac{\lambda}{c_2}\bm{\alpha}\mathbf{Q}^{-1}
\end{equation}
for Poisson intensity $\lambda$ of the claims arrival process\footnote{In the case when we have perturbed by the Brownian motion risk process, that is, $R(t)=R_1(t)+\sigma B(t)$, then
%\cite[eq. (4.2), p. 27]{Andreasbook2}
%\[\psi_2(z)=1-\varphi^\prime(0)W(x),\]
%where $\varphi(\theta)=\log E e^{\theta R(1)}=\frac{\sigma^2\theta^2}{2}+ c\theta -\lambda +\lambda (\theta \mathbf{I} -\mathbf{Q})^{-1}\mathbf{t}$ is a Laplace exponent of $R$,
%$\varphi^\prime(0)=c+\lambda {\bm{\alpha}}\mathbf{Q}^{-1}\mathbf{1} $ and $W$ is a scale function equal to
by \cite[Eq. (19)]{AAP}
\[\psi_2(z)=\sum_{j\in S}e^{\rho_j z}A_j,\]
%\MT{warto zapisac S w oznaczeniach z rozdzialu 4.}
where $\rho_j$ are distinct roots with strictly negative real part of the Cram\'er-Lundberg equation
\[\varphi(\rho)=0\]
for a Laplace exponent $\varphi(\theta)=\log E e^{\theta R(1)}=\frac{\sigma^2\theta^2}{2}+ c_1\theta -\lambda +\lambda (\theta \mathbf{I} -\mathbf{Q})^{-1}\mathbf{t}$ of $R$ and
\[A_j=\lim_{\theta \to \rho_j}\varphi(\theta)(\theta -\rho_j).\]
}.
\end{lemma}

Note that for the exponential distribution with parameter $\beta$, we have\footnote{See also Thm. 8.3.1, p. 340 of \cite{rolski}, see also Cor. 6.5.3, p. 252 of \cite{rolski}}
\begin{equation}\label{gamma}
\mathbf{Q}_+=-\gamma=-(\beta-\frac{\lambda}{c_2})=-\frac{\beta\theta_2}{1+\theta_2}
\end{equation}
and 
\begin{equation}\label{alphaplusexp}
\bm{\alpha}_+=-\frac{\lambda}{c_2\beta}.
\end{equation}

Hence, the main identity \eqref{masterformula} will give the expression for the two-dimension
ruin probability $\psi(u_1, u_2)$ as long as we identify
${\mathrm{P}_{(u_1,T)}}(\mathrm{d}z)$.

\section{Two-dimensional ruin for phase-type claims}
\label{sec:two}

Now, the numerical analysis of finding two-dimensional ruin probability $\psi(u_1, u_2)$ can be
done for general phase-type distributions. The matrix exponent appearing in \eqref{ruonphasetype} can be found by classical Jordan-type decomposition methods.

For some particular sub-families of phase-type distributions
the whole analysis can be further simplified.
We will now consider two such families of distributions based on \cite{BoxmaMandjes}.

For both families we assume the key condition that all solutions of a
Lundberg equation $\E e^{-sX} \; \frac{E(X)^{-1}}{E(X)^{-1}-(1+\theta)s} =1$ are real, hence the equation
\begin{equation}\label{Lundberg}
\bm{\alpha}(s\mathbf{I}-\mathbf{Q})^{-1}\mathbf{t}\;\frac{(\bm{\alpha}\mathbf{Q}^{-1}\mathbf{1})^{-1}}{(\bm{\alpha}\mathbf{Q}^{-1}\mathbf{1})^{-1}+(1+\theta)s} =1\quad \text{has all real roots $\kappa_i$ for $i=1, \ldots, m$},
\end{equation}
(with possible multiplicity, that is some of $\kappa_i$ might be
equal). We denote by $n_i$ the multiplicity of $\kappa_i$.
By Theorem 4.5 on p. 264 of \cite{asmussen} we know that this assumption is equivalent to
requirement that all eigenvalues of the matrix $\mathbf{Q}_+$ defined in \eqref{qplus} are real.
In other words this means that in the Jordan decomposition of this matrix
given by
\begin{equation}\label{Jordan}
\mathbf{Q}_+=\Delta{\rm diag}(K_i)\Delta^{-1}
\end{equation}
for matrix $\Delta$ with columns being right eigenvectors corresponding to $\kappa_i$,
there are no complex conjugate pairs in the set of solution $\kappa_i$ of \eqref{Lundberg}.
In \eqref{Jordan} $K_i$ is a Jordan block of size $n_i$ equal to
\begin{equation}\label{Jordnablock}
K_i:=
\left(\begin{array}{lllll}
\kappa_i&1&0&\dots&0\\
0&\kappa_i&1&\ldots&0\\
\ldots&\ldots&\ldots&\ldots&\ldots\\
0&0&0&\kappa_i&1\\
0&0&0&\ldots&\kappa_i
\end{array}
\right).
\end{equation}
Note that if all eigenvalues $\kappa_i$ are different, then $K_i=\kappa_i$ and
${\rm diag}(K_i)={\rm diag}(\kappa_i)$.
In particular, if $X$ has the exponential distribution with parameter $\beta$
then $m=1$ and $\kappa_1=-\gamma$.

The first class $\mathcal{M}$  corresponds to mixtures of independent exponentially distributed random variables
satisfying above condition \eqref{Lundberg}.
%with different intensities.
More precisely, for a given $k\in \mathbb{N}$,
\[f(x)=\sum_{i=1}^k \omega_i \beta_i e^{-\beta_ix}\]
where $\omega_i\geq 0$ with $\sum_{i=1}^k \omega_i=1$. % and $\beta_k\neq \beta_j$ for $j\neq l$.
The class $\mathcal{M}$ is suitable for representing random variables
with a squared coefficient of variation (scov) strictly larger than one as
one can find a distribution in $\mathcal{M}$ with the same moments; see
\cite[p. 359]{34} when $k=2$ for details.

Second class $\mathcal{S}$ corresponds to sums of independent exponentially distributed random variables
with parameters $\beta_i$ for $i=1,\ldots,k$ satisfying the condition \eqref{Lundberg}.
For this class one can match all finite moments for any distribution with scov strictly less than one.
Note that when all intensities of exponential distributions are equal than
resulting distribution has Erlang distribution with $k$ phases.
%We now argue that for distributions in $\mathcal{M}\cup\mathcal{S}$ the corresponding density can be written as
%a mixture of exponentials; this is a known property, but we include the underlying reasoning
%provides useful insights. As can be found in e.g. \cite[Prop. III.4.1]{asmussen2}, for any phase-type
%distribution the density can be directly evaluated from the matrix exponential pertaining to
%the so-called phase generator. For both $\mathcal{M}$ and $\mathcal{S}$ this phase generator is upper-triangular
%with $(\beta_1,\ldots,\beta_k)$ on the diagonal. This implies that (i) the eigenvalues can be read off from the diagonal
%and equal $\beta_1,\ldots,\beta_k$, and (ii) the eigenvectors can be evaluated at low computational cost
%by a recursive procedure.
%\[f(x)=\sum_{i=1}^k \omega_i \beta_i e^{-\beta_ix}\]
%for constants $\omega_i$ ($i=1,\ldots, k$) which are summing up to $1$, but in case of
%$\mathcal{S}$ not necessarily non-negative and $x > 0$.
The estimation of all the parameters of the distributions from class $\mathcal{M}\cup\mathcal{S}$
can be done via EM algorithm. 

From \eqref{ruonphasetype} and \eqref{Jordan}
we can conclude that
\begin{equation}\label{psi2onceagain}
\psi_2(z)=\bm{\alpha}_+\Delta{\rm diag}(e^{K_i z})\Delta^{-1}\mathbf{1}
\end{equation}
for
\begin{equation*}
e^{K_i z}=
\left(\begin{array}{lllll}
e^{\kappa_i z}&ze^{\kappa_iu}&\frac{1}{2!}z^2e^{\kappa_iz}&\ldots&\frac{1}{(n_i-1)!}z^{n_i-1}e^{\kappa_iz}\\
0&e^{\kappa_iz}&z e^{\kappa_iz}&\ldots&\frac{1}{(n_i-2)!}z^{n_i-2}e^{\kappa_iz}\\
\ldots&\ldots&\ldots&\ldots&\ldots\\
0&0&0&e^{\kappa_iz}&z e^{\kappa_iz}\\
0&0&0&\ldots&e^{\kappa_iz}\end{array}
\right).
\end{equation*}
Moreover, by safety loading condition $\theta_2>0$ and by considering large initial
reserves $z$ it follows that all $\kappa_i<0$.
Let $M\leq m$ be number of different solution of Lundberg equation \eqref{Lundberg}.
Then from \eqref{psi2onceagain}
it follows that
\begin{equation}\label{psi2final}
\psi_2(z)=\sum_{i=1}^M\sum_{j=1}^{n_i} \vartheta_{ij} z^{j-1} e^{\kappa_i z}\end{equation}
for some $\vartheta_{ij}$ ($i=1,\ldots,M$ and $j=1,\ldots, n_i$).
We recall that assumption that there are not
conjugate solutions of Lundberg equation \eqref{Lundberg}
is not always satisfied. As Dickson and Hipp \cite{DicksonHipp} show
if one takes symmetric mixture od Erlang$(2,1)$ and Erlang$(2,2)$, $\lambda=1$
and $c_2=4$ then $m=4$ and then $\bm{\alpha}=(1/2,0, 1/2, 0)$ and
\[\mathbf{Q} =
\left(\begin{array}{llll}
-1& 1& 0& 0\\
0 &-1& 0& 0\\
0 &0 &-2& 2 \\
0 &0 &0 &-2
\end{array}
\right).
\]
Moreover, we have then
\[\mathbf{Q}_+ =
\left(\begin{array}{llll}
-1& 1& 0& 0\\
1/8& -7/8& 1/16& 1/16\\
0& 0& -2& 2\\
1/4& 1/4& 1/8& -15/8
\end{array}
\right),
\]
and {$\bm{\alpha}_+=(1/8,1/8, 1/16, 1/16)$}.
Thus
\begin{align*}
\psi_2(z) &= 0.40026\exp(-0.51949 z)-0.04764\exp(-2.43637z)\\
&+0.02238\exp(-1.39707z)\cos(0.15311z)
- 0.21635\exp(-1.39707z)\sin(0.15311z)
\end{align*}
and it is not of the form of \eqref{psi2final}. In this case additionally $\cos({\Im} \kappa_i-\varrho_i)$ may appear
for $\varrho$ being a radial part of the constant in front of $e^{\kappa_i z}$.

Still, if one take the Erlang$(2, 1)$ of claim size distribution, then
for $\lambda= 1$ and $c_2 = 4$ we have $m=2$, $\bm{\alpha}=(1, 0)$ and
\begin{equation}\label{matrixQ}
\mathbf{Q} =
\left(\begin{array}{ll}
-1&1\\
0& -1
\end{array}
\right).
\end{equation}
Moreover, then
\begin{equation*}\mathbf{Q}_+ =
\left(\begin{array}{ll}
-1&1\\
1/4& -3/4
\end{array}
\right)
\end{equation*}
and {$\bm{\alpha}_+=(1/4,1/4)$}.
Thus
\begin{align}\label{Erlang}
\psi_2(z) = 0.55317 \exp(-0.35961z) - 0.05317 \exp(- 1.39039z)
\end{align}
and it is the form of \eqref{psi2final}.

Since $U_1(t)$ is a L\'evy process, hence by \cite{MeTomek}
we can introduce now the following exponential change of measure:
\begin{equation}\label{change}
\left.\frac{\mathrm{d} \mathrm{Q}_i}{\mathrm{d} \mathrm{P}}\right|_{\bm{\mathcal{F}_t}}=e^{\kappa_i (U_1(t)-u_1) -\varphi_it}
\end{equation}
for a natural filtration $\mathcal{F}_t$ of the process $\left(U_1(t), U_2(t)\right)_{t \geq 0}$
and
\begin{equation}\label{varphi}
\varphi_i=c_1\kappa_i +\lambda \left(\int_0^\infty e^{\kappa_iz}f(z)\mathrm{d}z-1\right)
= c_1\kappa_i +\lambda \left(\bm{\alpha}(\kappa_i\mathbf{I}-\mathbf{Q})^{-1}\mathbf{t}
-1\right).
\end{equation}
By $\mathrm{E}^{\mathrm{Q}_i}$ we will denote the expectation with respect to $\mathrm{Q}_i$.
Moreover, by Prop. 5.6 of \cite{MeTomek}, under probability measure $\mathrm{Q}_i$, process $U_1(t)$
equals $u_1+c_1t-\sum_{k=1}^{N(t)}X_k$ with $N(t)$ being the Poisson process with intensity
\begin{equation}\label{newlambda}\tilde{\lambda}_i=\lambda \int_0^\infty e^{\kappa_i z}f(z)\mathrm{d}z=
\lambda\bm{\alpha}(\kappa_i\mathbf{I}-\mathbf{Q})^{-1}\mathbf{t} \end{equation}

and generic claim size $X$ has new density function
\begin{equation}\label{densityQ}
\tilde{f}_i(x)=\frac{e^{\kappa_i x}f(x)}{\int_0^\infty e^{\kappa_i z}f(z)\mathrm{d}z}=
e^{\kappa_i z}f(z)\left[\bm{\alpha}(\kappa_i\mathbf{I}-\mathbf{Q})^{-1}\mathbf{t}\right]^{-1}=
{\bm{\alpha}}e^{\mathbf{Q}_ix}\mathbf{t}_i
\end{equation}
for
\[\mathbf{Q}_i:=\mathbf{Q}-\kappa_i\mathbf{I}\quad\text{and}\quad
\mathbf{t}_i:=-\mathbf{Q}_i\mathbf{1}\geq 0.\]

Note that from the representation \eqref{densityP} it follows that
$\tilde{f}_i(x)$ is again phase-type with generators $(\bm{\alpha}, \mathbf{Q}_i)$.
In particular, if $X$ has the exponential distribution with parameter $\beta$ then,
under probability measure $\mathrm{Q}_i$, the generic claim size $X$ has the exponential distribution with parameter
$-\mathbf{Q}_1=\beta+\kappa_1=
\beta-\gamma$ for $\gamma$ defined in \eqref{gamma}. Moreover, in this case $\tilde{\lambda}_1=\frac{\lambda \beta}{\beta-\gamma}$.
Finally, 
\begin{equation}\label{varphiexp}
\varphi_1=-c_1\gamma +\lambda\frac{\gamma}{\beta-\gamma}=-\gamma\left(c_1-\frac{\lambda}{\beta-\gamma}\right).
\end{equation}
Similarly, if $X$ has the Erlang distribution $(n,\beta)$ then,
under probability measure $\mathrm{Q}_i$, the generic claim size $X$ has Erlang distribution
$(n,\beta^Q_i)$ with
$\beta^Q=\beta+\kappa_i$.

The main results of this section is given by the following theorem.
\begin{theorem}
We have
\begin{align}
\psi(u_1,u_2)&=1\!-\!\mathrm{P}( \inf_{s\leq T}U_1(s)>0)\!\nonumber\\
&\qquad+\sum_{i=1}^M\sum_{j=1}^{n_i}\vartheta_{ij}\;
\frac{\partial^{j-1}}{\partial \kappa_i^{j-1}}\left\{
e^{\varphi_i T}e^{\kappa_i u_1}
\mathrm{Q}_i\left(\inf_{s\leq T}U_1(s)>0\right)\right\},\label{masterformula4}
\end{align}
where $\varphi_i$ is defined in \eqref{varphi},
$\vartheta_{ij}$ are defined in \eqref{psi2final} via \eqref{psi2onceagain} and $\kappa_i$ ($i=1,\ldots,M$)
solve Lundberg equation \eqref{Lundberg}.
For $j=1$ the partial derivative $\frac{\partial^{j-1}}{\partial \kappa_i^{j-1}}$ is understood as not taken at all.
\end{theorem}
\begin{proof}
From \eqref{masterformula} and \eqref{psi2final} we have
\begin{align}
\label{masterformula3}
\psi(u_1,u_2)=1\!-\!\mathrm{P}( \inf_{s\leq T}U_1(s)>0)\!+\sum_{i=1}^M\sum_{j=1}^{n_i}\vartheta_{ij}\!\!\!\mathrm{}\int_{0}^{\infty}\!\!\!\!
z^{j-1}e^{\kappa_i z}{\mathrm{P}_{(u_1,T)}}(\mathrm{d}z).
\end{align}
Moreover, note that by the definition of measure $\mathrm{Q}_i$ given in \eqref{change} it follows that
\begin{align*}\int_{0}^{\infty}\!\!\!\!
z^{j-1}e^{\kappa_i z}{\mathrm{P}_{(u_1,T)}}(\mathrm{d}z) &=
\frac{\partial^{j-1}}{\partial \kappa_i^{j-1}}\int_{0}^{\infty}\!\!\!\!
e^{\kappa_i z}{\mathrm{P}_{(u_1,T)}}(\mathrm{d}z)\\& =
\frac{\partial^{j-1}}{\partial \kappa_i^{j-1}}
\mathrm{E}\left[e^{\kappa_i U_1(T)}; \inf_{s\leq T}U_1(s)>0\right]\\&=
\frac{\partial^{j-1}}{\partial \kappa_i^{j-1}}\left\{
e^{\varphi_i T}e^{\kappa_i u_1}
\mathrm{Q}_i\left(\inf_{s\leq T}U_1(s)>0\right)\right\}.
\end{align*}
\end{proof}
We denote by $\mathbf{a}_i$ the $i$th component of a vector $\mathbf{a}$.
\begin{corollary}\label{maincor}
Let us now assume that all $m$ solutions $\kappa_i$ of Lundberg equation \eqref{Lundberg} are different.
Then $M=m$, $n_i=1$ and
\begin{align}
\psi(u_1,u_2)&=1\!-\!\mathrm{P}( \inf_{s\leq T}U_1(s)>0)\!\nonumber\\
&\qquad+\sum_{i=1}^m\vartheta_{i}\;
e^{\varphi_i T}e^{\kappa_i u_1}
\mathrm{Q}_i\left(\inf_{s\leq T}U_1(s)>0\right),\label{finalformula}
\end{align}
where
\begin{equation}\label{vartheta}
\vartheta_{i}:=
\left(\bm{\alpha}_+\Delta\right)_i\left(\Delta^{-1}\mathbf{1}\right)_i
\end{equation}
for $\Delta$ defined via Jordan decomposition $\mathbf{Q}_+=\Delta{\rm diag}(\kappa_i)\Delta^{-1}$ with
$\bm{\alpha}_+$ and $\mathbf{Q}_+$ defined \eqref{qplus}
and  $\varphi_i$ is defined in \eqref{varphi}.
Moreover, under $\mathrm{P}$ and $\mathrm{Q}_i$, the claim size density is given by \eqref{densityP}
and \eqref{densityQ}, respectively.

\end{corollary}

If the claim size has a exponential distribution with parameter $\beta>0$ then by \eqref{alphaplusexp}
$\vartheta_{1}=\bm{\alpha}_+=-\frac{\lambda}{c_2\beta}$.
Thus by \eqref{varphiexp}
\begin{align}
\psi(u_1,u_2)&=1\!-\!\mathrm{P}( \inf_{s\leq T}U_1(s)>0)\!\nonumber\\
&\qquad-\frac{\lambda}{c_2\beta} 
e^{-\gamma\left(c_1-\frac{\lambda}{\beta-\gamma}\right) T}e^{-\gamma u_1}
\mathrm{Q}_i\left(\inf_{s\leq T}U_1(s)>0\right).
\end{align}

Note that $\mathrm{P}( \inf_{s\leq T} U_1(s)>0)$ and $\mathrm{Q}_i( \inf_{s\leq T}U_1(s)>0)$
%can be found using Lemma \ref{Michanetallemma}. \MT{$\leftarrow$ chcemy skasowac rozdzial 4} It 
can be calculated by using numerical procedures, see e.g. \cite{Stanford}.
Another approach is related with
the power series expansion which is done the claim size distributions of mixed Erlang type in
\cite{389} and \cite{322}; see also \cite{310,893}. An alternative very accurate numerical method is to
randomize the time horizon $T$.
The detailed numerical analysis in some special cases of phase-type distribution of claim sizes
and other comments will be subject of next section.

\subsection{Ruin probability for the mixture of two exponentials}

\begin{example} \rm
In this part we establish a ruin probability for the model (\ref{2risk}) assuming that the claims follows a mixture of exponential distributions. For the simplicity in presentation of results we investigate a mixture of two exponential distributions given by positive weights $\omega_1, \omega_2$ that $\omega_1+\omega_2=1$ and means $\beta^{-1}_1$ and $\beta^{-1}_2$, respectively. Our result can be easily extended to the case where the mixture consists of finite number of exponential distributions.
\end{example}

$$\psi(u_1,u_2)=\Prob(\inf\{t\geq 0 : R_2(t) < 0\}<\infty).$$
Clearly, this observation reduces the two-dimensional case to the purely one-dimensional problem which has been solved analytically for the class of phase-type claims \cite{asmussen, rolski}.

\subsection{Ruin probability for the Erlang}

\begin{example}\rm
If $X$ has Erlang $(2,1)$ law then
from \eqref{Erlang} and \eqref{psi2final} it follows that
$\kappa_1=0.35961$, $\kappa_2=1.39039$. Then in the next step
from \eqref{varphi} we find
$\varphi_i$ ($i=1,2$) where $\bm{\alpha}=(1, 0)$, $\lambda=1$
and matrix $\mathbf{Q}$ is given in \eqref{matrixQ}.
Then from Corollary \ref{maincor} we can conclude that in this case
\begin{align}
\psi(u_1,u_2)&=1\!-\!\mathrm{P}( \inf_{s\leq T}U_1(s)>0)\!\nonumber\\
&\qquad+0.55317\;
e^{\varphi_1 T}
\mathrm{Q}_1\left(\inf_{s\leq T}U_1(s)>0\right)\\
&\qquad- 0.05317\;
e^{\varphi_2 T}
\mathrm{Q}_2\left(\inf_{s\leq T}U_1(s)>0\right).
\end{align}
where under measures $\mathrm{P}$, $\mathrm{Q}_1$, $\mathrm{Q}_2$ the risk process
$U_1$ has premium $c_1$ and claim size Erlang distributed with parameters
$(2,1)$, $(2, 1.35961)$, $(2, 2.39039)$, respectively.
Note that $\mathrm{P}( \inf_{s\leq T}U_1(s)>0)=\mathrm{P}(\tau>T)$ and
 $\mathrm{Q}_i( \inf_{s\leq T}U_1(s)>0)=\mathrm{Q}_i(\tau>T)$ ($i=1,2$)
 for the ruin time $\tau$ for the risk process $U_1(t)$
 and $T$ given in \eqref{T}.
To find this quantity, it is enough to find the density $w(u_1,t)$ of the
ruin time $\tau$. This can be done using \cite[page 58]{322}.
\end{example}

%\subsection{Ruin probability for the mixture of Erlangs}

\section{Numerical analysis for phase-type distributions}
\label{sec:sim}

We analyse now real-world loss data describing liability insurance claims obtained from a Polish insurance company in the years 2004-2012. The first step is to prepare the data so that the claim amounts are discounted at the same moment and aggregated on the single claim basis.  Analysis of the empirical claim amount distribution reveals two claims that deviate from the rest of the sample. These two claims constitute 5.31\% of all claims. For the purpose of this study they were excluded as outliers. The final sample consists of 542 payments and is shown in Figure \ref{fig:dane}.

\begin{figure}[!ht]
		\centering	
	\includegraphics[height=6cm]{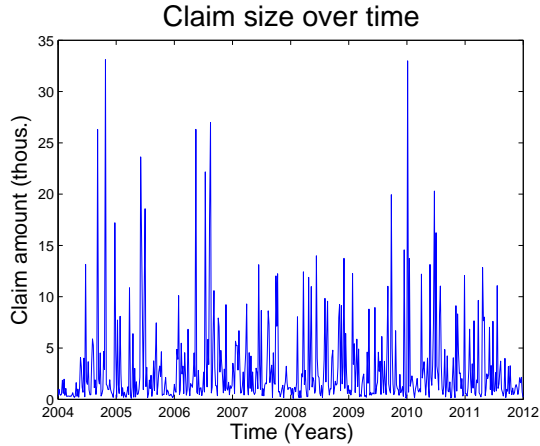}	
	\caption{Third-part liability insurance data from a Polish insurance company in the years 2010-2012.}
	\label{fig:dane}	
\end{figure}

The following distributions were taken into account to describe the claim amounts: exponential, mixture of exponentials, Erlang and mixture of Erlangs.   To check the goodness of fit we consider four test statistics based on the distance between the empirical and fitted distribution function.

The first considered test statistic is the classical Kolmogorov--Smirnov statistic $D$ based on the supremum norm defined as:
$$D= \sup_{x} |F_{n}(x)-F(x)|.$$
A similar statistic is the Kuiper $V$: 
$$	V= D^{+}+D^{-},$$
where $D^{+} = \sup_{x}\{F_{n}(x)-F(x)\}$ and  $D^{-} = \sup_{x}\{F(x)-F_{n}(x)\}$.

We will also use statistics calculated on the basis of quadratic norm, namely Cramer---von Mises $W^2$ and Anderson and Darling $A^2$ statistics:
$$
	W^2= n\int_{-\infty}^{\infty} (F_{n}(x)-F(x))^2 dF(x)$$
and 
	$$A^2= n\int_{-\infty}^{\infty} (F_{n}(x)-F(x))^2[F(x)(1-F(x))]^{-1} dF(x).
    $$

The former statistic puts more weight on observations in the tails of the distribution and is one of the most powerful statistical tests for detecting most departures from normality, cf.  \cite{dagste86}.

In order to estimate the parameters of the distributions we apply the estimation method based on minimising $A^2$ statistics. To calculate $p$-values for the studied tests we follow the Monte Carlo simulation algorithm  described in \cite{buretal11}.
The results of parameter estimation and hypothesis testing for the data are presented in Table~\ref{tab:tests}.

To calculate $p$-values for the studied tests we follow the Monte Carlo simulation algorithm  described in \cite{buretal11}.
The results of the parameter estimation and hypothesis testing for the third-party liability insurance claims amount are presented in Table~\ref{tab:tests}.
We decided not to include exponential distribution in the table since  for the Erlang distribution the coefficient $\alpha = 1$, which means that Erlang distribution simplifies to the exponential.
 
  \begin{center}
\begin{table}[!ht]
 \begin{center}
	\caption{Parameter estimates and test statistics for the Polish third-party liability insurance data. The corresponding $p$-values based on 1000 simulated samples are given in parentheses.}
    \label{tab:tests}
	\begin{tabular}{|lccc|} 		
		\hline
	Distribution & Mixture of exps &  Erlang & Mixture of Erlangs  \\
  \hline\hline
	Parameters  &  &  & \\  
	$\alpha = $ & $(0.1984,0.8016) $ &  $(1) $ &  $(0.8673,0.1327,0)$ \\
	 $Q =$  & $10^{-4} \left[\begin{array}{cc}
  -1.14  &  0    \\
  0 &  -5.90 
 \end{array}\right]$ & $10^{-4} \left[\begin{array}{c}
  4
 \end{array}\right]$ &$10^{-4} \left[\begin{array}{ccc}
  -6  &  0  &  0 \\
  0 &  -2 & 2 \\
  0 &  0 & -2 
 \end{array}\right]$  \\
 \hline
	Test results & $D = 0.0663 $ &  $D = 0.0967 $ & $D =  0.0661 $ \\
	& $(<0.005) $ & $ (<0.005) $ & $(<0.005)$  \\ 
	 & $ V=0.1229  $ & $ V=0.1662   $ & $ V=0.1223  $  \\
	& $ (<0.005)$ & $(<0.005)$ & $(<0.005)$  \\
	& $W^2= 0.4940$  & $W^2= 1.3207 $ & $W^2= 0.4869 $\\
	& $(0.02)$&   $ (<0.005) $ & $(<0.005)$   \\ 
	&  $A^2 = 4.3677  $ &  $A^2 = 12.4655 $ &  $A^2 = 4.3115  $  \\
	& $(<0.005)$ & $(<0.005)$  & $(<0.005)$  \\ 
		\hline		
	\end{tabular}
  \end{center}
\end{table}
 \end{center}

 % Results for ERLANG distribution:                   
% ===================================================
% Estimators derived by minimizing A2 statistic:     
% lambda=     0.0004
% k=          1.0000
% --------------------------------------
% EDF Statistics for these parameters:  
% --------------------------------------
% EDF Statistics for these parameters:  
% Kolmogorov-Smirnov D= 0.0967 p-value<0.005
% Kuiper             V= 0.1662 p-value<0.005
% Cramer-von Mises  W2= 1.3207 p-value<0.005
% Anderson-Darling  A2= 12.4655 p-value<0.005

	%	$\beta_1 = 1.1419 \cdot 10^{-4} $ & $k = 0.8211$ & $ \lambda_1= 0.0021, \lambda_2 = 0.002 $  
  %      \\ &  $\beta_2 = 5.9003 \cdot 10^{-4} $ &   &  $k_1 = 2.0623, k_2 = 1.1989 $ \\
 
 Unfortunately, neither of the proposed distributions passes the tests. However, we can see that the mixture of Erlang distributions with parameters $\alpha = (0.8673,0.1327,0)$ and $10^{-4} \left[\begin{array}{ccc}
  -6  &  0  &  0 \\
  0 &  -2 & 2 \\
  0 &  0 & -2 
 \end{array}\right]$  has the best results in terms of test statistics (the statistic values are the lowest). 
 
 We also check the quality of fit graphically by comparing the cumulative empirical and fitted distribution functions, see Figure \ref{fig:dists}. In addition, a histogram is plotted with theoretical probability functions corresponding to the fitted distributions.
 The illustrations suggest that mixtures of exponential and Erlang distributions are best fitted to the data.

\begin{figure*}[t!]
    \centering
    \begin{subfigure}[t]{0.55\textwidth}
        \centering
        \includegraphics[height=2.1in]{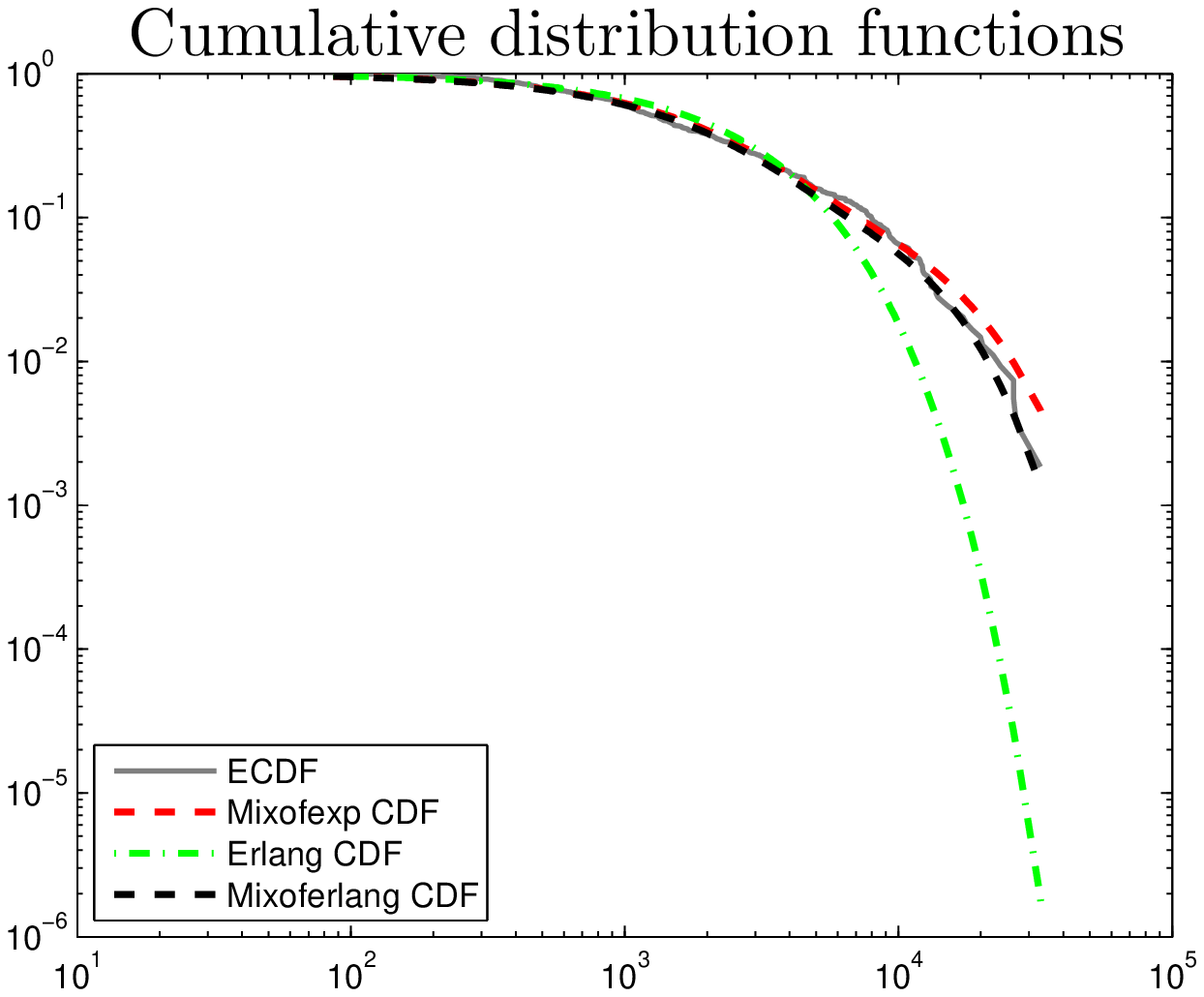}
%        \caption{Empirical cumulative distribution function (ECDF) and fitted cumulative distribution functions for fitted distributions.}
    \end{subfigure}%
    ~ 
    \begin{subfigure}[t]{0.45\textwidth}
        \centering
        \includegraphics[height=2.1in]{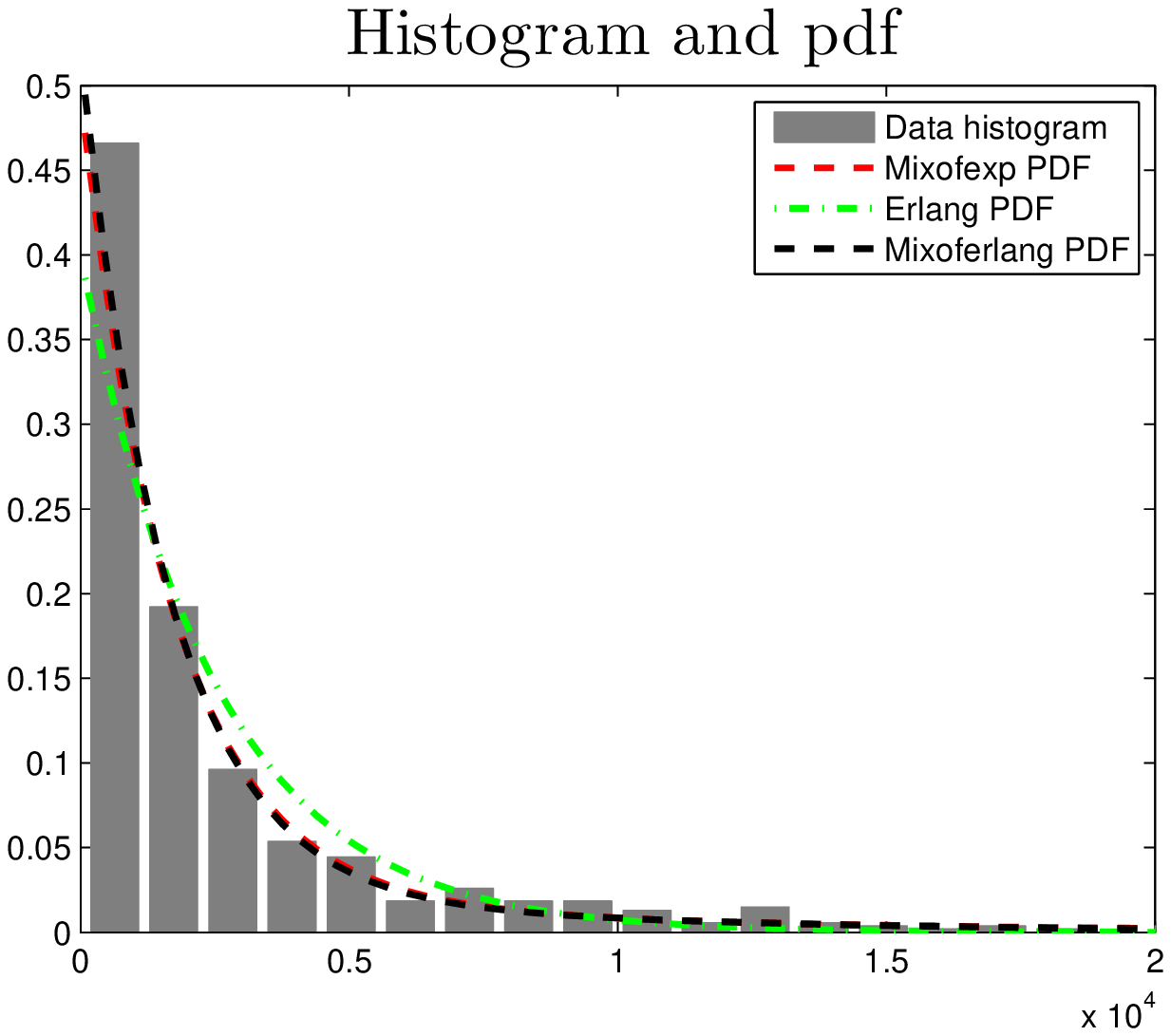}
%        \caption{Histogram and probability density functions (PDFs) for fitted distributions.}
    \end{subfigure}
    \caption{(Left panel) 
    Empirical cumulative distribution function (ECDF) and fitted cumulative distribution functions for fitted distributions. (Right panel) 
    Histogram and probability density functions (PDFs) for fitted distributions.}
    \label{fig:dists}
\end{figure*}

Now, we identify the claim counting process.
 Firstly, we determine the number of claims in subsequent months. They are depicted in Figure \ref{fig:liczba2}.

\begin{figure}[!ht]
	\centering	
	\includegraphics[height=6cm]{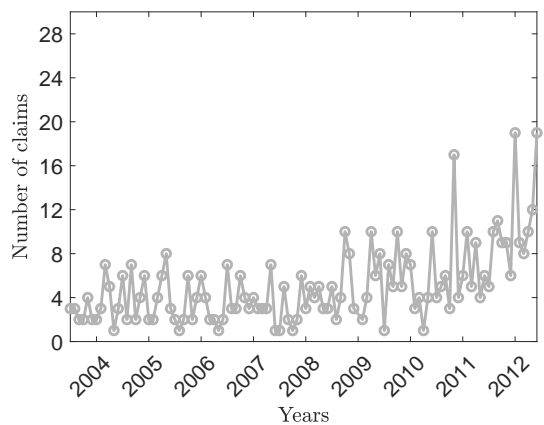}	
	\caption{Number of claims in months for the analysed data.}
	\label{fig:liczba2}	
\end{figure}

We do not observe seasonality, however, one can observe an increase in the number of claims in subsequent years. That is why we decided to apply a non-homogeneous Poisson process with $\lambda(t)$ being polynomial or exponential function. To find a proper intensity function, we fit polynomial or exponential functions to the aggregated number of claims. Parameters of the functions are estimated by minimisation of mean-squared error (the error is calculated with respect to the mean value function of the non-homogeneous Poisson process). In Figure \ref{fig:poissonf} we present the graphical comparison of the aggregate number of claims 
with the mean value function for analysed forms of intensity function.

\begin{figure}[!ht]
	\centering	
	\includegraphics[height=6cm]{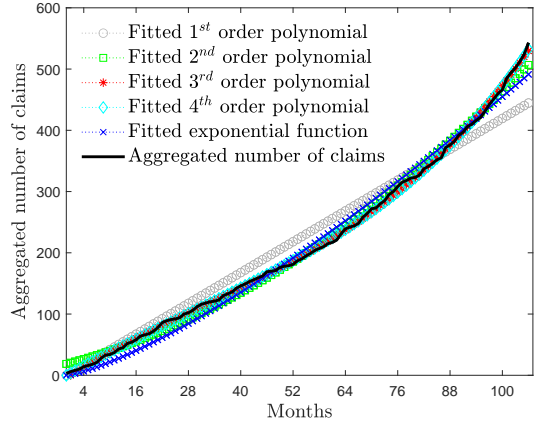}	
	\caption{Aggregate number of claims 
and the mean value functions corresponding to the fitted intensity functions.}
	\label{fig:poissonf}	
\end{figure} 

We notice that a first-degree polynomial is clearly worse suited to the aggregate number of claims. To verify the quality of fit of the intensity functions which are higher order polynomials or exponential functions, we determine the mean-squared error of the considered functions and present the results in Table~\ref{tab:msepoiss}.

\begin{center}
\begin{table}[!ht]
 \begin{center}
	\caption{Mean-squared error for the considered intensity functions
$\lambda(t)$ }
    \label{tab:msepoiss}
	\begin{tabular}{lccccc}
		\hline
	& \multicolumn{4}{c}{Order of the polynomial}&  \\
		\hline
	$\lambda(t)$ & $1^{st}$  & $2^{st}$   & $3^{rd}$  &    $4^{st}$  & Exp. function \\
  \hline
	MSE  & 22.23 &  8.52 &  3.16 & 3.09 &13.20
 \\  
 \hline
	\end{tabular}
  \end{center}
\end{table}
 \end{center} 
 
 Based on the mean-square errors, we choose the $3^{rd}$ order polynomial. For the higher order, the gain is negligible, and for the exponential function it even increases. Therefore, for the considered data, we select the non-homogeneous Poisson process with the intensity function:

\begin{equation*}
    \lambda(t) = 0.04t^3+4.54t^2+0.0004t-4.38.
\end{equation*}

\subsection{Probability of ruin for different scenarios}

We now calculate the ruin probability values using the formulas derived in Section \ref{sec:two}. The obtained results are compared to the probability of ruin calculated with the use of empirical distribution function (non-parametric bootstrap). In Figure \ref{fig:ruin} 
we can see the results for the fitted mixture of exponentials, Erlang and mixture of Erlangs, and the $90\%$ confidence interval obtained by the non-parametric bootstrap.

\begin{figure}[!ht]
		\centering	
	\includegraphics[height=6cm]{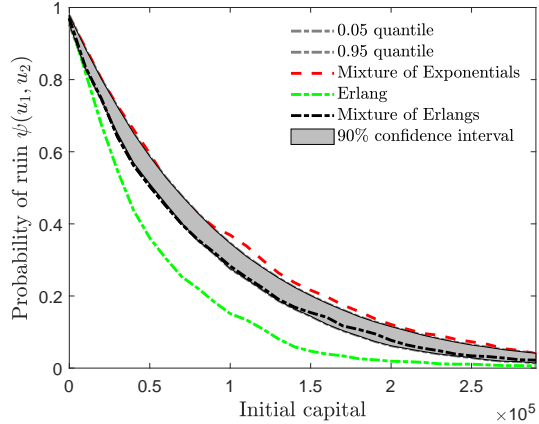}	
	\caption{Ruin probabilities with respect to the initial capital $u$ for the fitted mixture of exponential, Erlang mixture of Erlang distributions along with the 90$\%$ confidence interval created by means of non-parametric bootstrap of the loss data.}
	\label{fig:ruin}	
\end{figure}

We can clearly observe that only for the mixture of Erlang distributions the obtained values lie in the area between the quantiles of order 0.05 and 0.95, which proves the goodness of fit for this distribution. The probability of ruin calculated for the mixture of exponential distributions seems overestimated and for Erlang it is heavily underestimated. The analysis carried out leads to the conclusion that the Erlang mixture is the best suited for the data.

\section{Conclusions}
\label{sec:con}

In this paper, the problem of ruin probability in the case of a two-dimensional risk process for general phase-type claim amounts is investigated. The considered risk process assumes that both premiums and claims are divided between two lines in the same fixed proportions. Such a system can describe the capitals of the insurer and reinsurer under the quota share contract or two lines of business of the insurance company where the claims split on a pro rata basis. 

Our main findings are based on the purely stochastic arguments.
%and allows us to establish a ruin probability for the risk process with exponential claims and reinterpret similar system with claims that follows the mixture of exponential laws. 
Our main technique is the change of the measure that allows us to express the Laplace transform as the ruin probability of some modified risk process.
We derived infinite-time ruin probability formulas for general phase-type distributions and present specialised results for the mixture of exponential, Erlang and mixture of Erlang distributions.  

In order to illustrate presented results, we considered loss data from a Polish insurance company. The data contained claim amounts resulting from third-party liability insurance between 2004 and 2012.
We fitted a non-homogeneous Poisson process to the claim counting process and considered exponential,  mixture of exponential, Erlang and mixture of Erlang distributions as candidates to describe the claim amount sequence.
We performed statistical tests based on empirical cumulative distribution function and analysed the right tails of the fitted distribution. Finally, we calculate the ruin probability values for the considered distributions and compared them with the ruin probabilities obtained from the empirical distribution function. The analyses show that the model based on the mixture of two Erlang distributions is the best fitted to the data which illustrates the usefulness of phase type distributions in the context of the risk assessment.

%\section{Acknowledgment}

\bibliographystyle{elsarticle-harv}
\bibliography{biblio}
\end{document}